\documentclass[12pt, reqno]{amsart}
\usepackage{amsmath, amsthm, amscd, amsfonts, amssymb, graphicx, color}
\usepackage{float}

\usepackage{graphicx}
\usepackage{tikz}
\usetikzlibrary{shapes,arrows}
\newcommand{\GraphNodeDistance}{0.8 cm}
\newcommand{\GraphNodeSize}{2 pt}
\newcommand{\GraphInnerSep}{1 pt}
\newcommand{\GraphLineWidth}{0.6 pt}

\newcommand{\B}{\mathcal{B}}
\newcommand{\T}{\mathcal{T}}

\usepackage[mathscr]{eucal}
\newcommand{\W}{\mathscr{W}}
\newcommand{\w}{\mathfrak{W}}

\newcommand{\C}{\mathscr{C}}

\newtheorem{theorem}{Theorem}[section]
\newtheorem{lemma}[theorem]{Lemma}
\newtheorem{proposition}[theorem]{Proposition}
\newtheorem{corollary}[theorem]{Corollary}

\theoremstyle{definition}
\newtheorem{definition}[theorem]{Definition}

\newtheorem{remark}[theorem]{Remark}
\numberwithin{equation}{section}

\begin{document}

\author[M. Alinejad and K. Khashyarmanesh]
{ M.~Alinejad and K.~Khashyarmanesh}

\address{Department of Pure Mathematics, Ferdowsi University, Mashhad, Iran.}
\email{\textcolor[rgb]{0.00,0.00,0.84}{mohsenalinejad96@gmail.com}}
\email{\textcolor[rgb]{0.00,0.00,0.84}{khashyar@ipm.ir}}

\title[Counting short cycles of (c,d)-regular bipartite graphs]
{Counting short cycles of (c,d)-regular bipartite graphs}
\maketitle

\begin{abstract}
Recently, working on the Tanner graph which represents a low density parity check (LDPC) code 
becomes an interesting research subject. Finding the number of short 
cycles of Tanner graphs motivated Blake and Lin to investigate the multiplicity of cycles of length 
girth in bi-regular bipartite graphs, by using the spectrum and degree distribution of the graph. Although there were many algorithms 
to find the number of cycles, they preferred to investigate in a computational way. Dehghan and Banihashemi counted 
the number of cycles of length $g+2$ and $g+4,$ where $G$ is a bi-regular bipartite graph and 
$g$ is the length of the girth $G.$ But they just proposed a descriptive technique to compute the multiplicity of cycles 
of length less than $2g$ for bi-regular bipartite graphs. In this paper, 
we find the number of cycles of length less than $2g$ by using spectrum and degree distribution of  bi-regular bipartite graphs such 
that the formula depends only on the partitions of positive integers and the number of closed cycle-free walks from 
a variable (resp. check) vertex in $ \B_{c,d}$ and
$\T_{c,d}$ $($resp. $\T_{d,c}),$ which are known.
\end{abstract}

\section{Introduction}
Low density parity check (LDPC) codes are linear codes with error performance near to the Shannon limit can 
represent as Tanner graphs, which proposed the first time by Michael Tanner. It is shown that the 
structure of Tanner graphs, in particular, distribution and number of short cycles, affect the 
error efficiency of LDPC codes (see \cite{hu, karimi, tanner, xiao}).
Performance of the LDPC codes persuade many researchers to investigate 
on cycles of Tanner graphs. In \cite{chugg, mao}, are shown that
for finding LDPC codes with a good performance, we should focus on the graphs that the number of short cycles is not so many. 
It is proved that regularity of a graph has a significant impact on LDPC codes  \cite{richard}. Dehghan and 
Banihashemi in \cite{bani}, studied cycle distribution of random bipartite graphs.

Counting the number of cycles in a general graph is known to be NP-hrad \cite{flum}. The complexity of the problem led the researcher 
to use recursion methods and algorithms to compute the number of cycles in bipartite graphs. For instance, in \cite{hal}, 
Halford and Chugg 
proposed a recursive algorithm for counting the cycles of length at most $g+6.$
In \cite{kar}, Karimi and Banihashemi presented 
an algorithm to compute the number of cycles of length less than $g.$
Recently, Blake and Lin suggested a new way, independent from algorithms and complicated methods, to compute the number of cycles 
by using spectrum and degree distribution of bipartite graphs \cite{blake}.

For a given graph G, the \textit{adjacency matrix} $A = [a_{ij} ]$ of $G$ is defined such that
$a_{ij} = 1,$ if $ij \in E(G)$, and $a_{ij} = 0,$ if $ij \notin
E(G).$ The \textit{spectrum} of a graph $G,$ denoted by $\lbrace \lambda_i \rbrace ,$ is the multiset of 
 eigenvalues of adjacency matrix $A.$ Since there exists a 
close relationship between the number of walks of arbitrary length and powers of matrix $A,$ the spectrum of 
$G$ is more useful to find the number of cycles. 
Blake and Lin in \cite{blake} found the 
number of short cycles of length $g$ in bi-regular bipartite graphs without using complicated algorithms. They were hoping 
this new method guided the 
researchers to find the number of cycles of length greater than $g.$ In \cite{dehghan}, Dehghan and Banihashemi  
determined the exact 
number of cycles of length $g+2$ and $g+4$ in a bi-regular bipartite graph. 
In addition, 
by contradiction examples, they showed that the spectrum and degree distribution conditions are not enough to find
the number of 
cycles of length $i$ for a bi-regular 
bipartite graph, where $i \geq 2g.$
They also mentioned some facts for 
the number of cycles of length less than $2g,$ but they did not proposed a formula to compute the number of cycles. By using 
the eigenvalues and degree sequence of bi-regular bipartite graphs, we present a new way to enumerate the number of cycles of length 
less than $2g$ in bi-regular bipartite graphs. 

In Section $2$ of the paper, we present some definitions and preliminaries which we need through this paper. In Section 
$3,$ by using the partitions of positive integer numbers, we find the number of closed walks with a cycle in 
bi-regular bipartite graphs in which initial vertex is 
in cycle. In Section $4,$ similar to section $3,$ we investigate the number of closed walks that consist a cycle and initial 
vertex is out of cycle. Finally, from the results of sections $3$ and $4,$ we determine the number of closed walks with cycle.
Since the number of closed cycle-free walks in bi-regular bipartite graphs 
specified in \cite{blake}, we can express the number of cycles of length less 
than $2g.$   

\section{Preliminaries and Notations}
For any graph $G$, we denote the set of all vertices and edges of $G$ by $V(G)$ and $E(G)$, respectively. 
For two vertices $u, v \in V(G),$ we denote $ u \sim v$ or $ uv $ for brevity, if $u$ and $v$ are adjacent.
The \textit{degree} of a vertex $v\in V(G)$, denoted by $d(v)$, 
is the number of adjacent vertices of $v$. 
A \textit{walk} $\W$ is a sequence of the vertices $v_1, v_2, \ldots, v_{k+1}$ such that $v_{j} v_{j+1} \in E(G),$ for 
$1 \leq j \leq k.$ In this case, $v_j$ is called the $j$-th vertex of $\W$ and
the length of $\W$ is defined as the number of edges of $\W$ and is denoted by $\ell(\W).$ 
We call $v_1$ and $v_{k+1}$ the \textit{initial} and \textit{terminal} vertex of $\W,$ respectively.
For integers $j$ and $s,$ a walk $\W'=v_{j}, \ldots, v_{v_{j+s}}$ is a \textit{subwalk} of $\W
=v_1, v_2, \ldots, v_{k+1},$ if 
$1 \leq j < k+1$ and $1 < j+s \leq k+1.$ A walk is called a \textit{closed walk} if the initial and terminal vertex are the same. 
A closed \textit{cycle-free} walk is a closed walk with no cycles. A closed walk 
$\W$ which is not a cycle
 is called a \textit{closed walk with cycle}, if the 
induced subgraph on the edges of the closed walk has at least one cycle. For brevity, we denote the 
closed walk with cycle by CWWC. 
 If the vertices of a walk are distinct,
 then 
a walk or closed walk is called \textit{path} and \textit{cycle}, respectively.
 For $u,v\in V(G)$, 
$d(u,v)$ denotes the length of the shortest path between $u$ and
$v$. If there is no path between $u$ and $v$, then we define $d(u,v)=\infty$. For a graph $G,$ length of shortest cycle is called 
\textit{girth}, and is denoted 
by $g.$ For $j \geq g,$ the number of cycles of length $j$ is denoted by $N_{j}.$ 

Graph $G$ is called \textit{bipartite}, if $V(G)$ can be partitioned into two sets $U$ and $V$ such that 
if $uv \in E(G),$ then $u$ and $v$ belong to different sets. A graph $G$ is called nonbipartite, if $G$ is not bipartite.
 If the degree of vertices $U$ and $V$ are $c$ and $d,$ respectively, 
then $G$ is called $(c,d)$-regular bipartite graph, and is denoted by $\B_{c,d}$. In this case, 
we assume that
$|U|=n$ and $|V|=m.$ For a bipartite graph $G=U \cup V,$ the $m \times n$\textit{ parity check } matrix $H(G)=[h_{ij}]$ 
defined in which $h_{ij}=1,$ if $ij \in E(G), $ and $h_{ij}=0,$ otherwise. Clearly, $H(G)$ constructs a linear code $C(G).$ 
In this case, $G$ is called the \textit{ Tanner graph} of $C(G).$ 
We denote $b_{c,2k} $ and $a_{c,2k}$ $ ($resp. $b_{d,2k}$ and $a_{d,2k})$ as the number of closed cycle-free and 
return once closed cycle-free walks of length $2k$ with initial vertex of degree $c $ in $\B_{c,d}($resp. $d)$. For 
$k=0,$ we assume that $b_{c,0}=b_{d,0}=1.$

Graph $G$ is \textit{connected}, if there is a path between each two vertices of $G.$ A connected graph with no cycle is called 
\textit{tree}. For graph $\B_{c,d}$, the \textit{related tree} of $\B_{c,d},$ denoted by $\T_{c,d} $ $($resp. $ 
\T_{d,c}),$
is defined as the rooted tree with root vertex of degree $c-1$ $($resp. $d-1)$ and vertices of consecutive levels have alternating degrees 
$d$ $($resp. $c)$ and $c $ $($resp. $d).$ 
 In addition, we denote 
$t_{c,2k} $ and $s_{c,2k}$ $ ($resp. $t_{d,2k}$ and $s_{d,2k})$ as the number of closed cycle-free and 
return once closed cycle-free walks of length $2k$ in $\T_{c,d}$ $($resp. $\T_{d,c})$ with initial vertex of degree $c $ $($resp. $d)$. 
For $k=0,$ we assume that $t_{c,0}=t_{d,0}=1.$

Clearly, the adjacency matrix $A$ is real and symmetric, and so the eigenvalues of $G$ are real. 
Moreover, it is known that 
if $\lbrace \lambda_i \rbrace$ is the spectrum of $G,$ then $\lbrace \lambda^k_i \rbrace$ is the eigenvalue of 
$A^k,$ for a positive integer 
$k.$ 
For a matrix $A,$ $tr(A)$ is defined as the summation of diagonal entries of $A.$ 
The following proposition play an important role in the enumerating the number of cycles.
\begin{proposition} \cite [Proposition $1.3.4$] {cvet}  \label{s11}
If $A$ is the adjacency matrix of a graph, then $(i,j)$-entry $a^k_{ij}$ of the matrix $A^{k}$ is equal to the 
number of walks of length $k$ that start at vertex $i$ and end at vertex $j.$
\end{proposition}
Since
$tr(A)$ is equal to the summation of eigenvalues of $A,$ Proposition \ref{s11} implies that 
the number of closed walks of length $k$ equals the summation of eigenvalues of $A^{k}.$ 
The Following theorem shows the difference of the spectrum of bipartite and the spectrum nonbipartite graphs.
\begin{theorem}  \cite [Theorem $3.2.3$] {cvet}
A graph $G$ is bipartite if and only if its spectrum is symmetric with respect to the origin.
\end{theorem}
 Now, we express the following theorem 
from \cite{ dehghan}, which shows the number of cycles of length $i,$ where $i<2g.$ 
\begin{theorem} \cite [Theorem $1$] {dehghan} \label{s16}
For a $(c,d)$-regular bipartite graph $\B_{c,d},$ the number of cycles of length $i$ is equal to: 
$$N_i= [\sum\limits_{j=1}^{|V(G)|} {\lambda^i_j}- \Omega_i(c,d,\B_{c,d})- \Psi_i(c,d,\B_{c,d})]/ 2i ,$$ 
where $\lbrace \lambda ^i_j\rbrace$ is the spectrum of $\B_{c,d},$ and $\Omega_i(c,d,\B_{c,d}) $ and $ \Psi_i(c,d,\B_{c,d})$ are the 
number of closed cycle-free walks of length $i$ and closed walks with cycle of length $i$ in $G,$ respectively. 
\end{theorem}
In \cite{blake}, Blake and Lin have already found the value $\Omega_i(c,d,\B_{c,d}) $ and showed that 
$$ \Omega_i(c,d,\B_{c,d}) = n b_{c,i}+m b_{d,i}.$$

In this work, we compute the number of closed walks of length 
$i+2k,$ $1 \leq k < g-\frac{i}{2},$ which contain a cycle of length $i, i <2g.$ For each walk of $\B_{c,d}$ 
we can consider a direction. 
Let $\W$ be an arbitrary closed walk with cycle of 
length $i.$
By passing the walk sequence of $\W,$ if we traverse clockwise in the cycle, 
then we define the direction of $\mathscr{W}$ is clockwise. 
Otherwise, we define the direction of $\W$ is counterclockwise. A CWWC walk with direction counterclockwise is denoted by 
$CWDCC.$
 Now, suppose that $\C$ is 
a cycle of length $i$ with vertices 
$v_j, 0 \leq j \leq i-1($see Fig. $1).$ Throughout this paper, indices of vertices of $\C$ are taken modulo $i$ and $d(v_0)=c.$
\begin{figure}[H]
\centering
\begin{tikzpicture}[node distance=\GraphNodeDistance, >=stealth',
 minimum size=\GraphNodeSize, inner sep=\GraphInnerSep, line width=\GraphLineWidth]
\tikzstyle{init} = [pin edge={to-, thin, white}]
\tikzstyle{place}=[circle, draw ,thick,fill=black]
\tikzstyle{label}=[circle , minimum size=1 pt,thick]
\node [place] (v0) at (-3,-5) {};
\node [place] (v1) at (-4,-5) {};
\node [place] (v2) at (-4.5,-4) {};
\node [place] (v3) at (-2.5,-4) {};
\node [place] (v4) at (-4,-3) {};
\node [place] (v5) at (-3,-3) {};

\node [label] at (-2.8,-5.8) [left] {Fig 1. };

\draw[-] (v0) -- (v1);
\draw[-] (v0) -- (v3);
\draw[-] (v1) -- (v2);
\draw[-] (v4) -- (v5);
\draw[-] (v2) -- (v4);

\node [label] at (-4,-3) [left] { $v_0$};
\node [label] at (-4,-5) [left] { $v_2$};
\node [label] at (-4.5,-4) [left] { $v_1$};
\node [label] at (-3,-3) [right] { $v_{i-1}$};
\node [label] (hdot1) at (-2.75,-3.5) {$.$};
\node [label] (hdot1) at (-2.85,-3.3) {$.$};
\node [label] (hdot1) at (-2.65,-3.7) {$.$};
\end{tikzpicture}
\end{figure}
\section{Closed walks in $(c,d)$-regular graphs with initial vertex in cycle}
\begin{definition}
For integer $k,$ where $1 \leq k < g-\frac{i}{2},$ the set of closed walks of length $i+2k$ with cycle of length 
$i$ such that the initial vertex is in cycle or out of cycle is denoted by 
$\Phi(i,2k)$ or $\Lambda(i,2k),$ respectively. Therefore, 
$$\Psi_{i}(c,d,\B_{c,d})= 
\sum\limits_{k_0+k_1= g-\frac{i}{2}-1}^{} {\Big{(}|\Phi(g+2k_0,2k_1)|+|\Lambda(g+2k_0,2k_1)|\Big{)}}. 
$$
\begin{definition}
Suppose that $\W$ is a CWDCC 
with walk sequence $u_1, \ldots, u_{k+1}.$ 
For a positive integer $j$ less than $k+1,$ a closed cycle-free walk $\W'$ with initial vertex $u_j$ is called \textit{backward} 
respect to 
the $\W,$ if $u_j u_{j+1} \notin E(\W').$ The walks $\W^{\ast}$ and $\W^{\ast \ast}$ are defined as the CWDCC of Fig. $1$ and 
Fig. $2$ with starting vertex $v_0$ and $z_0,$ respectively.
\end{definition}
\end{definition}
In this section, we investigate the number of walks of $\Phi(i,2k)$ that consists cycle $\C.$ We 
first determine the number of CWDCCs of $\Phi(i,2k)$ with cycle $\C$ and initial vertex $v_0,$ denoted
$\Phi_{v_0}(i,2k).$ 
\begin{remark} \label{s14}
Suppose that $\W \in \Phi_{v_0}(i,2k).$ Then represent $\W$ as a sequence 
of vertices. Denote the first $v_0$ which appears in the sequence
by $v^{(1)}_0$, and the first $v_j$ that appears after 
$v^{(1)}_{j-1}$ denote by $v^{(1)}_{j},$
 for every $j,$ $0< j \leq i.$ Since $\W$ is a CWDCC, we deduce that the vertex before $v^{(1)}_{j}$ in 
the sequence is $v_{j-1},$ where $0< j \leq i.$ Hence, we can represent  $\W$ as follow:
$$v^{(1)}_0, \ldots ,v_0, v^{(1)}_{1}, \ldots ,v_{1}, v^{(1)}_{2}, \ldots , v_{i-1}, v^{(1)}_{i}, \ldots ,v_0.$$
Now, define closed subwalk $\W_j$ of $\W$ with initial vertex $v^{(1)}_{j}$ and 
terminal vertex $v^{(1)}_{j+1},$ for each $j$ with $ 0\leq j \leq i-1.$ 
Moreover, the walk 
with initial vertex $v^{(1)}_{i}$ and 
terminal vertex $v_0($The last vertex of $\W )$ is denoted by $\W_{i}.$ 
Thus, we can represent $\W$ uniquely as follow:
$$\W= \W_0 \W_1 \cdots \W_{i}.$$
Now, put $2s'_i=\ell(\W_{i}),$ and $2s'_j=\ell(\W_{j})-1,$ for each $j$ with 
$ 0 \leq j \leq i.$ By simple computing, we have 

$$s'_0+ \ldots + s'_{i}=\frac{1}{2}  \Big{(} \sum\limits_{j = 0}^{i} {\ell(\W_{j})}- \sum\limits_{j = 0}^{i-1} {1} \Big{)} $$

$$= 
\frac{1}{2}( \ell(\W)- i)$$ 
$$=
\frac{1}{2}( i+2k- i )=k.
$$
\end{remark}

\begin{definition}
We denote $\Phi_{v_0}(2s_0, \ldots,2s_{i})$ for 
the walks $ \W \in \Phi_{v_0}(i,2k)$ such that if $\W$ represent uniquely as
$\W= \W_0 \W_1 \cdots \W_{i},$ then $2s_j=\ell(\W_j)-1$ for $0 \leq j \leq i-1,$ and $2s_{i}= \ell(\W_i).$
\end{definition}
By the Remark \ref{s11}, we have following result. 
\begin{corollary} \label{s12}
For a positive integer $k$ with $1 \leq k < g-\frac{i}{2},$ we have:
$$\Phi_{v_0}(i,2k)=\bigcup\limits_{{s_0} +  \ldots  + {s_{i }} = k}^{} {{\Phi _{{v_0}}}({2s_0}, \ldots ,{2s_{i }})}. $$
\end{corollary}

It is easy to see that the number of CWDCCs with initial vertex $v_{2j},$ $0 \leq j \leq \frac{i}{2}-1,$ is
equal. Similar result satisfies for $v_{2j+1},$ where $0 \leq j \leq \frac{i}{2}-1.$ 
 Since the number of vertices of degree $c$ and $d$ are $\frac{i}{2}$ and each CWWC has two directions, we
have following result. 
\begin{corollary}
The 
number of walks of $\Phi(i,2k)$ with cycle $\C$ is equal to:
$$i \Big{(}\sum\limits_{{s_0} +  \ldots  + {\rm{ }}{s_{i }} = k} {|{\Phi _{{v_0}}}(2{s_0}, \ldots ,2{s_{i}})| +  
|{\Phi _{{v_1}}}(2{s_0}, \ldots ,2{s_{i}})|} \Big{)}.
$$
\end{corollary}
\begin{lemma} \label{s6}
If $s_0+\ldots+s_{i}=k,$ then 
$$|{\Phi _{{v_0}}}({2s_0}, \ldots ,{2s_{i }})|= 
b_{ c,2s_{ i } } \prod\limits_{j =0}^{\frac{i}{2}-1} {t_{ c,2s_{2 j} } t_{ d,2s_{2 j+1} },} $$
and 
$$|{\Phi _{{v_1}}}({2s_0}, \ldots ,{2s_{i }})|= 
b_{ d,2s_{ i } } \prod\limits_{j =0}^{\frac{i}{2}-1} {t_{ d,2s_{2 j} } t_{ c,2s_{2 j+1} }}. $$

\end{lemma}
\begin{proof}
By the definition of ${\Phi _{{v_0}}}({2s_0}, \ldots ,{2s_{i }}),$ we know that 
$2s_j$ is the length of the closed cycle-free walk with 
initial vertex $v_j$ 
in which it is backward respect to $\W^{\ast},$ for each $j,$ $0 \leq j \leq i-1.$ 
Hence, Dependes on whether the degree of $v_j$ is $c$ or $d,$ the number of walks of length $2s_j$ 
with this 
condition is equal to $t_{c, 2s_j}$ or $t_{d, 2s_j},$ respectively.
On the other hand, $2s_i$ is the length of closed cycle-free walks with 
cycle and initial vertex $v_0.$ Hence, the number of these 
walks equals to $b_{c,2s_{i}}$ or $b_{d,2s_{i}}.$ Similar result is satisfied for 
$|{\Phi _{{v_1}}}({2s_0}, \ldots ,{2s_{i }})|.$ 
\end{proof}
Since we computed the value $|\Phi(i,2k)|$ for specific cycle $\C$ and we have $N_i$ cycles of length $i,$ we have following result. 
\begin{corollary}  \label{s4}
For integer $k,$ $1 \leq k < g-\frac{i}{2} ,$ we have 
$$|\Phi(i,2k)|= i N_{i} \sum\limits_{{s_0} +  \ldots  + {\rm{ }}{s_{i }} = k} {(b_{ c,2s_{ i } }+ b_{ d,2s_{ i } })}  
\prod\limits_{j =0}^{\frac{i}{2}-1} {t_{ d,2s_{2 j} } t_{ c,2s_{2 j+1} }} .
$$
\end{corollary} 
From \cite{blake}, we observe that $a_{c,2l} = c$ $ t_{d,2l-2}, $ and $a_{d,2l} = d $ $t_{c,2l-2}$ 
for $l<g.$ Thus, we can rewrite the equation of Corollary
\ref{s4} by using just $a_{c,2j},b_{c,2j}, a_{d,2j}$ and $b_{d,2j},$ for $1 \leq j \leq k.$ 
\section{Closed walks in $(c,d)$-regular graphs with initial vertex out of cycle}
In this section, we investigate on the size of
$\Lambda(i,2k).$ In the following definition, we classify the CWDCC's of $\C$ with initial vertex out of cycle.

\begin{definition}
Let $l$ and $j$ be integers with $1 \leq l \leq k$ and $0 \leq j < i.$ Then $\w
({v_j},l)$ is defined as the set of  
the CWDCCs with cycle $\C$ and initial vertex $z$ such that satisfy in the following conditions:
\begin{itemize}
\item[(i)]
$z \in V(G-\C)$. 
\item[(ii)]
$d(z,v_j)=l.$

\item[(iii)] 
If $\W \in \w({v_j},l),$ then the first vertex of $\C$ that appears in the walk sequence of $\W$ is $v_j.$ 
\end{itemize} 
\end{definition}
The set of initial vertices of $\w({v_j},l)$ is denoted by $N_{l}(v_j),$ where $1 \leq l \leq k.$ 
It is not difficult to see that if $d(v_j)=c,$ then $|N_{l}(v_j)|$ is equal to: 
$$(c-2)(d-1)^{\lceil\frac{l-1}{2}\rceil}(c-1)^{\lfloor \frac{l-1}{2}\rfloor}.$$ 
Otherwise, 
$$(d-2)(c-1)^{\lceil\frac{l-1}{2}\rceil}(d-1)^{\lfloor \frac{l-1}{2}\rfloor}.$$

First, we find the number of walks $ \W \in \Lambda(i,2k)$ with cycle $\C$ such that $v_0$ is 
the first vertex of $\C$ appears in the walk sequence $\W.$
For $z_0 \in N_{l}(v_0),$ since there is not a cycle of 
length $2l,$ we have a unique path of length $l$ between $z_0$ and $v_0, $ say $P_{z_0v_0}.$ We denote this unique 
path by
$ 
z_0 \sim z_1 \sim z_2 \ldots z_{l-1} \sim v_0 ($see Fig. $2).$ The set of walks of $\w({v_0},l) $ with initial vertex 
$z_0$ is denoted by $\Lambda^{z_0}_{v_0}(l).$ 
\begin{figure}[H]
\centering
\begin{tikzpicture}[node distance=\GraphNodeDistance, >=stealth',
 minimum size=\GraphNodeSize, inner sep=\GraphInnerSep, line width=\GraphLineWidth]
\tikzstyle{init} = [pin edge={to-, thin, white}]
\tikzstyle{place}=[circle, draw ,thick,fill=black]
\tikzstyle{label}=[circle , minimum size=1 pt,thick]
\node [place] (v0) at (-3,-5) {};
\node [place] (v1) at (-4,-5) {};
\node [place] (v2) at (-4.5,-4) {};
\node [place] (v3) at (-2.5,-4) {};
\node [place] (v4) at (-4,-3) {};
\node [place] (v5) at (-3,-3) {};
\node [place] (v6) at (-7.6,-3) {};
\node [place] (v7) at (-6.7,-3) {};
\node [place] (v8) at (-5.8,-3) {};
\node [place] (v9) at (-4.9,-3) {};

\node [label] at (-4,-5.7) [left] {Fig. $2$. };

\draw[-] (v5) -- (v8);
\draw[-] (v0) -- (v1);
\draw[-] (v6) -- (v7);
\draw[-] (v8) -- (v9);
\draw[-] (v0) -- (v3);
\draw[-] (v1) -- (v2);
\draw[-] (v4) -- (v5);
\draw[-] (v2) -- (v4);

\node [label] at (-3.7,-2.7) [left] { $v_0$};
\node [label] at (-4.5,-4) [left] { $v_1$};
\node [label] at (-4.9,-2.7) [] { $z_{l-1}$};
\node [label] at (-5.8,-2.7) [] { $z_{l-2}$};
\node [label] at (-6.7,-2.7) [] { $z_{1}$};
\node [label] at (-7.6,-2.7) [] { $z_{0}$};
\node [label] at (-3,-3) [right] { $v_{i-1}$};
\node [label] (hdot1) at (-2.75,-3.5) {$.$};
\node [label] (hdot1) at (-2.85,-3.3) {$.$};
\node [label] (hdot1) at (-2.65,-3.7) {$.$};
\node [label] (hdot1) at (-6.2,-3) {$\ldots$};

\end{tikzpicture}
\end{figure}
\begin{remark} \label{s13}
Suppose that $\W \in \Lambda^{z_0}_{v_0}(l).$ 
Denote the initial and terminal vertex of $\W$ by $u_0$ and $u_{i+2l+1},$ respectively. Denote 
the second vertex of $\W^{\ast \ast}$ that 
appears in $\W$  
after the $u_0$ by $u_1,$ and the third vertex of 
$\W^{\ast \ast}$ which 
appears in $\W$  
after the $u_1$ by $u_2.$
Continuing in this way, denote $j$-th vertex of $\W^{\ast \ast}$ 
that appears after $u_{j-2},$ by $u_{j-1},$ for $2 \leq j \leq i+2l+1.$ Now, define the subwalks of $\W$ as follow. For each 
$j, 0 \leq j \leq i+2l,$ define subwalk $\W_j$ of $\W$ with initial and terminal vertex $u_j$ and $u_{j+1}.$ Hence, we can represent  
$\W$ uniquely as bellow:
$$\W= \W_0 \cdots \W_{i+2l}.$$
Similar to the Remark \ref{s11}, assume that $2s'_{i+2l}=\ell(\W_{i+2l}),$ and $2s'_j=\ell(\W_{j})-1,$ for each $j$ with 
$ 0 \leq j \leq i+2l-1.$ By simple computing, we have 

$$s'_0+ \ldots + s'_{i+2l}=\frac{1}{2}  \Big{(} \sum\limits_{j = 0}^{i+2l} {\ell(\W_{j})} - \sum\limits_{j = 0}^{i+2l-1} {1} \Big{)} $$
$$ \  \  \  \  \  \  \  \  \  \  \  \  \    \  \   \  \  \  \   \  \   \ \   \  \  \  \   = 
\frac{1}{2} \Big{(}\ell(\W)-(i+2l) \Big{)}$$ 
$$  \  \  \  \  \  \  \  \  \  \  \  \  \    \  \   \  \  \  \   \  \   \ \   \  \  \  \   \  \   =
\frac{1}{2}\Big{(} (i+2k)- (i+2l)\Big{)}=k-l.
$$
\end{remark} 
\begin{definition}
We denote $\Lambda^{{z_0}}_{v_0}({2s_0}, \ldots ,{2s_{i+2l }})$ for 
the walks of $\W \in \Lambda^{z_0}_{v_0}(l)$ such that if we represent $\W$ uniquely as
$\W= \W_0 \W_1 \cdots \W_{i+2l},$ then $2s_j=\ell(\W_j)-1,$ for $0 \leq j < i+2l,$ and $2s_{i+2l}= \ell(\W_{i+2l}).$
\end{definition}
Since for each walk of $ \Lambda^{z_0}_{v_0}(l),$ there is a unique representation, by Remark \ref{s13}, we have 
the following result.
\begin{corollary} \label{s15}
For a positive integer $l$ with $1 \leq l \leq k,$ we have
$$\Lambda^{z_0}_{v_0}(l)=\bigcup\limits_{{s_0} +  \ldots  + {s_{i+2l}} = k-l}^{} 
{{\Lambda^{{z_0}}_{v_0}}({2s_0}, \ldots ,{2s_{i+2l }})}.$$
Moreover, 

$$ |\Lambda^{z_0}_{v_0}(l)|= \sum \limits_{{s_0} +  \ldots  + {s_{i+2l}} = k-l}
{|{\Lambda^{{z_0}}_{v_0}}({2s_0}, \ldots ,{2s_{i+2l }})|} .$$
\end{corollary}

Now, we can state the following lemma.
\begin{lemma} \label{s9}
For $s_0+\ldots+s_{i+2l}=k-l.$ If $l$ is even or odd, then we have
$$|{\Lambda^{{z_0}}_{v_0}}({2s_0}, \ldots ,{2s_{i+2l }})| =b_{ c,2s_{ i +2l} } 
\prod\limits_{j =0}^{\frac{i}{2}+l-1} {t_{ c,2s_{2 j} } t_{ d,2s_{2 j+1} },}$$ 
and 
$$ |{\Lambda^{{z_0}}_{v_0}}({2s_0}, \ldots ,{2s_{i+2l }})| =
b_{ d,2s_{ i +2l} } \prod\limits_{j =0}^{\frac{i}{2}+l-1} {t_{ d,2s_{2 j} } t_{ c,2s_{2 j+1} },}$$ 
respectively.
\end{lemma}
\begin{proof}
Let $\W \in {\Lambda^{{z_0}}_{v_0}}({2s_0}, \ldots ,{2s_{i+2l }}).$ 
Then Remark \ref{s13} implies that there are the unique subwalks $\W_j$'s such that $\W= \W_1 \cdots \W_{i+2l}$ 
and $2s_j=\ell(\W_j)-1,$ for $0 \leq j < i+2l$ and $2s_{i+2l}=\ell(\W_{i+2l}).$ Hence, we can deduce that  
for $0 \leq j < i+2l,$ $2s_j$ is the length of backward closed cycle-free walk 
respect to the $\W^{\ast \ast}$
with $(j+1)$-th vertex of $\W^{\ast \ast}$ as initial. 
Hence, the number of cycle-free walks in this case is equal to $t_{c, 2s_j}$ or $t_{d, 2s_j},$ depends on the degree
initial vertex is $c$ or $d,$ respectively.
In addition, since $2s_{i+2l}=\ell(\W_{i+2l}),$ we conclude that $2s_{i+2l}$ 
is the length of closed cycle-free walk with 
initial vertex $z_0.$ Therefore, the number of cycle-free walks in this case 
is $b_{c,2s_{i+2l}}$ or $b_{d,2s_{i+2l}},$ if $l$ is even or odd, respectively.
Now, it is enough to show that the walks which computed in the right side of the equation are elements of 
${\Lambda^{{z_0}}_{v_0}}({2s_0}, \ldots ,{2s_{i+2l }}). $ Since the initial vertex of the walks is $z_0,$ we just 
check that $v_0$ is the first vertex of $\C$ that appears in our enumeration. By contradiction assume that 
$\W' \in {\Lambda^{{z_0}}_{v_0}}({2s_0}, \ldots ,{2s_{i+2l }})$ and 
vertex $v_j, 0<j \leq i-1,$ 
appears before the $v_0$ in $\W'.$ Since $d(z_0,v_0)=l,$ we have $d(z_0,v_j) \leq k-l.$ Hence, there 
is a new cycle of length at most 
$$d(z_0,v_0)+ d(z_0,v_j)+ d(v_0,v_j)\leq l+(k-l)+\frac{i}{2}.$$ 
Since $i+2k <2g,$ the length of cycle is less than $g,$ which is a contradiction.
\end{proof}
\begin{lemma} \label{s3}
Let $l$ be a positive integer with $1 \leq l \leq k$ and $z_0 \in N_{l}(v_0).$ Then
$$|\w({v_0},l) |=  |N_{l}(v_0)| |\Lambda^{z_0}_{v_0}(l)|$$  
\end{lemma}
\begin{proof}
Suppose that $w_0 \in N_{l}(v_0)$ and $w_0 \neq z_0.$ Since $d(w_0)=d(z_0),$ $d(w_0, v_0)=d(z_0, v_0)=l$ and 
$\Lambda^{z_0}_{v_0}(l) \cap \Lambda^{w_0}_{v_0}(l)= \emptyset,$ we observe that 
$|\Lambda^{z_0}_{v_0}(l)|= |\Lambda^{w_0}_{v_0}(l)|.$ Hence, we only
compute $|\Lambda^{z_0}_{v_0}(l)| $ and finally multiply by 
$|N_{l}(v_0)|.$  
\end{proof}
\begin{lemma}
The number of walks $\W \in \Lambda(i,2k)$ with cycle $\C$ 
such that $v_0$ is the first vertex of $\C$ that appears in $\W$ is
$$ \sum\limits_{l=1}^{k} {|\w({v_0},l) |} $$  
\end{lemma}\label{s7}
\begin{proof}
Let $\W \in \Lambda(i,2k)$ with cycle $\C$ and initial vertex $z_0$ such that the first vertex of $\C$ that appears in 
$\W$ is $v_0.$ In this case, $1 \leq d(z_0,v_0) \leq k.$ Hence $\W \in \w({v_0},l),$ where 
$l= d(z_0,v_0).$ Since $\w({v_0},l) \cap \w({v_0},l')=\emptyset,$ for distinct $l$ and $l',$ the assertion holds.
\end{proof}
Note that we have $\frac{i}{2}$ vertices of degree $c$ and $\frac{i}{2}$ vertices of degree $d$ in $\C.$ 
In addition, for $0 \leq j \leq \frac{i}{2}-1,$ we have 
$$|\w(v_0, l)|=|\w(v_{2j}, l)|,$$ 
and 
$$|\w(v_1, l)|=|\w(v_{2j+1}, l)|.$$ 
Hence, we have the following consequence. 

\begin{corollary} \label{s2}
For a positive integer $k$ with $1 \leq k < g-\frac{i}{2},$ we have
$$|\Lambda(i,2k)|= i N_i \sum\limits_{l=1}^{k} {(|\w({v_0},l) |+|\w({v_1},l) |)}.$$
\end{corollary}
\begin{remark} \label{s8}
It is not difficult to see that to find $|\Lambda(i,2k)|,$ we may only calculate $ |\w({v_0},l) |$ for $1 \leq l \leq k.$ Because 
$\B_{c,d}$ is a bi-regular graph and we have similar result for $ |\w({v_1},l) |.$ Since we know the 
value $|N_{l}(v_0)|,$ it is enough to check the $|\Lambda^{z_0}_{v_0}(l)|$ to find $ |\w({v_0},l) |$
 for $1 \leq l \leq k.$ 
\end{remark}
For finding the number of cycles of length $i$ in $\B_{c,d},$ it is enough to investigate the value $\Psi _{i}(c,d,\B_{c,d}),$ by Theorem 
\ref{s16}.
In the next two theorems we enumerate $\Psi _{j}(c,d,\B_{c,d})$ for $j=g+2, g+4.$ Our proof is simpler 
than the proof of Theorem $2$ and Theorem $3$ in \cite{dehghan}. Finally, we compute  $\Psi _{g+6}(c,d,\B_{c,d})$.
\begin{theorem} \label{s5}
Let $G$ be a $(c, d)$-regular graph. Then
$$\Psi _{g+2}(c,d,\B_{c,d})= g N_{g} (g+2)(c+d-2).$$
\end{theorem}
\begin{proof}
To compute $\Psi _{g+2}(c,d,\B_{c,d}),$ we calculate the values $|\Phi(g,2k)|$ and $|\Lambda(g,2k)|,$ respectively. Since $k=1,$ 
Corollary \ref{s4} implies that  
$$|\Phi(g,2)|= g N_{g} \sum\limits_{{s_0} +  \ldots  + {\rm{ }}{s_{g }} = 1} {(b_{ c,2s_{ g } }+ b_{ d,2s_{ g } })}  
\prod\limits_{j =0}^{\frac{g}{2}-1} {t_{ d,2s_{2 j} } t_{ c,2s_{2 j+1} }} 
.$$ 
In this case, we have $g+1$ cases for $(s_0, \ldots, s_g).$ 
If $s_g=0,$ then $b_{c,0}+b_{d,0}=2.$ Since there are $\frac{g}{2}$ vertices of degree $c$ and $\frac{g}{2}$ vertices of degree $d,$ 
the number of CWDCCs in this case is $2 \frac{g}{2} 
t_{c,2} + 2\frac{g}{2} 
t_{d,2} .$ If $s_g =1,$ then the number of CWDCCs is $ b_{ c,2 }+$ $ b_{ d,2}.$ Therefore, 
$$|\Phi(g,2)|= g N_{g} \Big{(}g(t_{c,2}+t_{d,2}) +b_{c,2}+b_{d,2} \Big{)}.  \  \  \  \  \  \  \  (\ast)$$ 
By the corollary \ref{s2}, we have 
$$|\Lambda(g,2)|= g N_g \sum\limits_{l=1}^{1} {(|\w({v_0},l) |+|\w({v_1},l) |)}
$$ 
$$ \  \  \  \  \  \  \   = g N_g(|\w({v_0},1) |+|\w({v_1},1) |).$$  
It is enough to find the value $ |\Lambda^{z_0}_{v_0}(1)|,$ by Remark \ref{s8}. So
$$ |\Lambda^{z_0}_{v_0}(1)|= 
\sum\limits_{{s_0} +  \ldots  + {s_{g+2}} = 0}{|\Lambda^{{z_0}}_{v_0}({2s_0}, \ldots ,{2s_{g+2}})}|
=1.$$ 
Therefore, we deduce that $|\w({v_0},1) |=|N_1(v_0)|$ and  $|\w({v_1},1) |=|N_1(v_1)|.$ Thus,
$$|\Lambda(g,2)|=g N_g \Big{(}|N_1(v_0)|+|N_1(v_1)|\Big{)}. \  \  \  \  \  \  \  (\ast \ast)$$ 
By the equations $(\ast)$ and $(\ast \ast)$ we conclude that 
$$ \Psi _{g+2}(c,d,\B_{c,d})= g N_{g} \Big{(}g(t_{c,2}+t_{d,2}) +(b_{c,2}+b_{d,2}) +(|N_1(v_0)|+|N_1(v_1)|)\Big{)}.$$
\end{proof}
\begin{theorem} \label{s10}
Let $G$ be a $(c, d)$-regular graph. Then
$$\Psi _{g+4}(c,d,\B_{c,d})= (g+2) N_{g+2} (g+4)(c+d-2)$$ 
$$ \  \  \  \  \  \  \  \   \  \  \  \  \  \  \  \   \  \  \  \  \  \  \  \  \  \  \  \  \  + g N_g [g(c+d-2)^2+(c+d)(c+d-1)]$$
$$ \  \  \  \  \  \  \  \   \  \  \  \  \  \  \  \   \  \  \  \  \  \  \  \  \  \  \  \  \  \  \  \  \  \  \  \  \  \  \  \  \ 
\  \  \   \  \  \   \  \  \  \  \  \  \  \  +g N_g \Big{(}2 \left({\begin{array}{*{20}{c}}
{\frac{g}{2}}\\
{2}
\end{array}}\right) [(c-1)^2+(d-1)^2]+2(\frac{g}{2})^2 (c-1)(d-1) \Big{)}$$
$$ \  \  \  \  \  \  \  \  \  \  \  \  +g N_g \frac{g}{2}[(c+d)(c+d-2)]$$
$$ \  \  \  \  \  \  \  \   \  \  \  \  \  \  \  \   \  \  \  \  \  \  \  \  \  \  \  \  \  \  \  \  \  \  \  \  \  \  \  \  \ 
\  \  \   \  \  \  \  \  \   \  \  \  \  \   +g N_g \Big{(} (\frac{g}{2}+1) (c+d-2)(c+d-4)+(c-2 )(2d-1)+(d-2)(2c-1) \Big{)}.$$
\end{theorem}
\begin{proof}
To compute the number of CWWCs of the length $g+4,$ we need to know the number of cycles of length $g$ and $g+2$ which have 
already enumerated. Thus, we arise the following two cases:\\
\textbf{Case 1.} The closed walk of length $g+4$ contains a cycle of length $g.$ We first compute 
$|\Phi(g,4)|.$ In this case, $k=2$ and Corollary \ref{s4} implies that 
$$|\Phi(g,4)|= g N_{g} \sum\limits_{{s_0} +  \ldots  + {\rm{ }}{s_{g }} = 2} 
{(b_{ c,2s_{ g } }+ b_{ d,2s_{ g } })}  
\prod\limits_{j =0}^{\frac{g}{2}-1} {t_{ d,2s_{2 j} } t_{ c,2s_{2 j+1} }} 
.$$ 
Depending on $s_j$ is one or two, we consider the following subcases:\\ 
\textbf{Case 1.1.} Each $s_j$ is two or zero. In this case, if $s_g$ is zero or not, then the number of closed CWDCCs is equal to
$$2 \frac{g}{2} 
t_{c,4} + 2\frac{g}{2} 
t_{d,4},$$
or 
$$b_{c,4}+b_{d,4},$$ 
respectively. \\
\textbf{Case 1.2.}  Each $s_j$ is one or zero. Suppose that 
$s_g=0.$ In this sense, there are three cases to select two ones for $s_j$'s. If both vertices have the same degree, 
then the number of CWDCCs is 
$$  2 \left({\begin{array}{*{20}{c}}
{\frac{g}{2}}\\
{2}
\end{array}}\right) [t_{c,2}^2+t_{d,2}^2].
$$ 
If the degree of vertices are different, then the number of CWDCCs is
$$ 2(\frac{g}{2})^2 t_{c,2}t_{d,2}.$$
Now, suppose that $s_g \neq 0.$ In this case, $s_g=1$ and so there is another $j'$ such that $0 \leq j' < g$ and $s_j'=1.$ 
Hence, the number of CWDCCs is 
$$ \frac{g}{2}(b_{c,2}+b_{d,2})(t_{c,2}+t_{d,2}).$$
Thus, 
$$ \  \  \  \  \  \  \ |\Phi(g,4)|=g N_g \Big{(}g( 
t_{c,4} +  
t_{d,4})+ (b_{c,4}+b_{d,4})+\\
 2 \left({\begin{array}{*{20}{c}}
{\frac{g}{2}}\\
{2}

\end{array}}\right) [t_{c,2}^2+t_{d,2}^2]
$$ 

$$+
2(\frac{g}{2})^2 t_{c,2}t_{d,2}+
\frac{g}{2}(b_{c,2}+b_{d,2})(t_{c,2}+t_{d,2}) \Big{)}.
$$
Now, we want to find 
$|\Lambda(g,4)|. $
By the Corollary \ref{s2}, we have 
$$|\Lambda(g,4)|= g N_g \sum\limits_{l=1}^{2} {(|\w({v_0},l) |+|\w({v_1},l) |)}
$$ 
It is enough to find the values $|\Lambda^{z_0}_{v_0}(1)|$ and $|\Lambda^{z_0}_{v_0}(2)|,$ by Remark \ref{s8}. 
First consider $l=1$ and Lemma \ref{s9} implies that
$$ |\Lambda^{z_0}_{v_0}(1)|= 
\sum\limits_{{s_0} +  \ldots  + {s_{g+2}} = 1}{|{\Lambda^{{z_0}}_{v_0}}({2s_0}, \ldots ,{2s_{g+2 }})|}
$$ 
$$ \  \  \  \  \  \  \  \  \  \ =(\frac{g}{2}+1) t_{c,2}+ (\frac{g}{2}+1) t_{d,2}+ b_{d,2}.$$
Since $|\Lambda^{z_0}_{v_0}(2)|=1,$ we have 
$$ \  \  \  \  \  \  \  \  \  \  \   \ \  \ \\   
\  \  \  \  \  |\Lambda(g,4)|= g N_g \Big{(} (\frac{g}{2}+1) \big{(}|N_1(v_0)|+|N_1(v_1)| \big{)}(t_{c,2}+ t_{d,2})
$$
$$
\  \  \   \   \  \  \  \  \  \  \  \  \ \  \  \   \   \  \  \  \  \  \  \  \  \  \  \  \  \  \ +|N_1(v_0)|b_{d,2}+ 
|N_1(v_1)|b_{c,2}+N_{2}(v_0)+N_{2}(v_1) 
\Big{)}.$$ 
\textbf{Case 2.} 
The closed walk of length $g+4$ that contains a cycle of length $g+2.$ 
From the proof of Theorem \ref{s5}, the number of CWDCCs in this case is
$$ (g+2) N_{g+2} \Big{(}(g+2)(t_{c,2}+t_{d,2}) +(b_{c,2}+b_{d,2}) +(|N_1(v_0)|+|N_1(v_1)|)\Big{)}.$$
\end{proof}
Now, we find the value $\Psi _{g+6}(c,d,\B_{c,d})$ in the following three lemmas.
\begin{lemma}
Let $G$ be a $(c, d)$-regular graph. Then we have
$$
\  \  \  \ \  \  \  |\Phi(g,6)|=gN_g \Big{(}((c-1)^2+(d-1)^2) (c+d-2)\Big{[} g+
2 \left({\begin{array}{*{20}{c}}
{\frac{g}{2}}\\
{2}
\end{array}}\right)
+2 \left({\begin{array}{*{20}{c}}
{\frac{g}{2}}\\
{3}
\end{array}}\right)
\Big{]}$$
$$ \! \! \! \! \! \! \! \! \! \! \! \! \! \! \! \! \! \! \! \! \! + ((c-1)^2+(d-1)^2) (c+d)\Big{[} \left({\begin{array}{*{20}{c}}
{\frac{g}{2}}\\
{2}
\end{array}}\right)
+1
\Big{]}$$

$$ \  \  \  \  \  \ \  \  \  \  +
(c-1)(d-1)(c+d-2)
\Big{[}
3g+g^2+g \left({\begin{array}{*{20}{c}}
{\frac{g}{2}}\\
{2}
\end{array}}\right)
-2 \left({\begin{array}{*{20}{c}}
{\frac{g}{2}}\\
{3}
\end{array}}\right)
\Big{]}$$
$$ \! \! \! \! \! \! \! \! \! \! \! \! \! \!  + (c+d)
\Big{[}
(3cd-c-d)+\frac{g^2}{4} (c-1)(d-1) \Big{]} \Big{)},$$
and 
$$ \  \  \  \  \  \  \  \  \  \  \  \  \  \ \  \  \  \  \  \  \  \  \  \  |\Lambda(g,6)|= gN_g \Big{(}
(c+d-4) \Big{[}
(\frac{g}{2}+1)(c+d-2)^2
+ ((\frac{g}{2}+1)^2+1)(c-1)(d-1)$$
$$+ 
 \left({\begin{array}{*{20}{c}}
{\frac{g}{2}+1}\\
{2}
\end{array}}\right) 
((c-1)^2+(d-1)^2)
\Big{]}
 $$
$$ 
\  \  \  \  \  \  \  \  \  \ \  \  \ \ \  \  \  \  \  \  +(2cd-2c-2d) \Big{[} 
(c+d-1)+ (\frac{g}{2}+1)(c+d-2)
\Big{]}
$$
$$
\  \  \  \  \  \  \  \  \  \ \  \  \ \ \  \  \  \  \  \  \  \  +(\frac{g}{2}+2)(c+d-2)
\Big{[}
(c-2)(d-1)+(d-2)(c-1)
\Big{]}
$$
$$
\  \  +c(c-2)(d-1)+d(d-2)(c-1) \Big{)}.
$$

\end{lemma}
\begin{proof}
To enumerate the number of CWDCCs in this case, we first investigate $|\Phi(g,6)|.$ From the Corollary \ref{s4}, we have 
$$|\Phi(g,6)|= g N_{g} \sum\limits_{{s_0} +  \ldots  + {\rm{ }}{s_{g}} = 3} 
{(b_{ c,2s_{ g } }+ b_{ d,2s_{ g} })}  
\prod\limits_{j =0}^{\frac{g}{2}-1} {t_{ d,2s_{2 j} } t_{ c,2s_{2 j+1} }} 
.$$ 
In our proof, we avoid using $gN_g$ in our calculating. Now, consider the following three subcases:\\
\textbf{Case 1.} Each $s_j$ is zero or three. In the above summation, if $s_g=0,$ then the number of CWDCCs is 
$$2 \frac{g}{2} 
t_{c,6} + 2\frac{g}{2} 
t_{d,6} .$$
If $s_g =3,$ then the number of CWDCCs in this case is equal to:
$$  b_{c,6}+ b_{d,6}.$$\\
\textbf{Case 2.} Each $s_j$ is zero, one, or two. 
Depending on whether $s_{g}$ is zero or not, the number of CWDCCs is
$$ 2 \left({\begin{array}{*{20}{c}}
{\frac{g}{2}}\\
{2}
\end{array}}\right)(t_{c,2} t_{c,4}+ t_{d,2} t_{d,4}) +  2 \frac{g^2}{4} (t_{c,2} t_{d,4}+ t_{c,4} t_{d,2}),$$
and
$$ 
\frac{g}{2} 
\Big{(} (b_{c,4}+b_{d,4})(t_{c,2}+t_{d,2})+(b_{c,2}+b_{d,2}) (t_{c,4}+t_{d,4}) \Big{)},$$ 
respectively.
\\
\textbf{Case 3.} Each $s_j$ is zero or one. Again, suppose that $s_g=0,$ then we have three ones in distinct $s_j$'s. Hence, 
the number of CWDCCs is equal to
$$2\left({\begin{array}{*{20}{c}}
{\frac{g}{2}}\\
{3}
\end{array}}\right) (t^3_{c,2}+t^3_{d,2})+ 
2 \left({\begin{array}{*{20}{c}}
{\frac{g}{2}}\\
{2}
\end{array}}\right)
 \left({\begin{array}{*{20}{c}}
{\frac{g}{2}}\\
{1}
\end{array}}\right)
(t^2_{c,2} t_{d,2}+t_{c,2}t^2_{d,2}).
$$ 
If $s_g \neq 0,$ then $s_g=1$ and the number of CWDCCs is 
$$ (b_{c,2}+b_{d,2}) \Big{(}
\left({\begin{array}{*{20}{c}}
{\frac{g}{2}}\\
{2}
\end{array}}\right) (t^2_{c,2}+t^2_{d,2})
+
(\frac{g}{2})^2
t_{c,2} t_{d,2}
\Big{)}.$$
To complete the proof in this case, we find the value
$|\Lambda(g,6)|. $
By the corollary \ref{s2}, we have 
$$|\Lambda(g,6)|= g N_g \sum\limits_{l=1}^{3} {(|\w({v_0},l) |+|\w({v_1},l) |)}.
$$ 
It is enough to find $|\Lambda^{z_0}_{v_0}(l)|$ for $1 \leq l \leq 3.$ Hence, we consider the following three subcases:\\
\textbf{Case a.} Suppose that $l=1.$ Since $l$ is odd, Corollary \ref{s15} and Lemma \ref{s9} imply that 
$$ |\Lambda^{z_0}_{v_0}(1)|= 
\sum\limits_{{s_0} +  \ldots  + {s_{g+2}} = 2}{b_{ d,2s_{ g +2} } 
\prod\limits_{j =0}^{\frac{g}{2}} {t_{ d,2s_{2 j} } t_{ c,2s_{2 j+1} }}}.$$
If $s_j \in \lbrace 0,2 \rbrace,$ then the number of CWDCCs in this case is equal to:
$$(\frac{g}{2}+1)(t_{c,4}+t_{d,4})+b_{d,4}.$$ 
Now, suppose that $s_j \in \lbrace 0,1 \rbrace.$ Depending on whether $s_{g+2}$ is zero or not, the number of CWDCC's is 
$$ \left({\begin{array}{*{20}{c}}
{\frac{g}{2}+1}\\
{2}
\end{array}}\right) (t^2_{c,2}+t^2_{d,2})+
(\frac{g}{2}+1)^2 t_{c,2}t_{d,2}
,$$ 
and 
$$(\frac{g}{2}+1) b_{d,2} (t_{c,2}+t_{d,2}),$$ 
respectively. \\
\textbf{Case b.}
Suppose that $l=2.$ Since $l$ is even, we have 
$$ |\Lambda^{z_0}_{v_0}(2)|= 
\sum\limits_{{s_0} +  \ldots  + {s_{g+4}} = 1}{b_{ c,2s_{ g+4} } 
\prod\limits_{j =0}^{\frac{g}{2}+1} {t_{ c,2s_{2 j} } t_{ d,2s_{2 j+1} }}}.$$
Depending on $s_{j}$ is zero or not, we have the following number as the CWDCCs. 
$$ (\frac{g}{2}+2) (t_{c,2}+t_{d,2})+b_{c,2}.$$
\textbf{Case c.}
Assume that $l=3.$ Therefore, by Corollary \ref{s15} and Lemma \ref{s9} we have
$$ |\Lambda^{z_0}_{v_0}(3)|= 
\sum\limits_{{s_0} +  \ldots  + {s_{g+6}} = 0}{b_{ d,2s_{ g +6} } 
\prod\limits_{j =0}^{\frac{g}{2}+2} {t_{ d,2s_{2 j} } t_{ c,2s_{2 j+1} }}}=1.$$
\end{proof}
\begin{lemma}
Let $G$ be a $(c, d)$-regular graph. Then the number of CWWCs of length $g+6$ with cycle of length $g+2$ is equal to 
$$   \   \  \  \  \  \  \  \  \  \  \  \  \  \  \  \  (g+2) N_{g+2} [(g+2)(c+d-2)^2+(c+d)(c+d-1)]$$
$$ \  \  \  \  \  \  \  \   \  \  \  \  \  \  \  \   \  \  \  \  \  \  \  \  \  \  \  \   \  \  \  \  \  \  \   
+(g+2) N_{g+2} \Big{(}2 \left({\begin{array}{*{20}{c}}
{\frac{g+2}{2}}\\
{2}
\end{array}}\right) [(c-1)^2+(d-1)^2]+2(\frac{g+2}{2})^2 (c-1)(d-1) \Big{)}$$
$$  +(g+2) ( \frac{g+2}{2})N_{g+2}[(c+d)(c+d-2)]$$
$$ \  \  \  \  \  \  \  \   \  \  \  \  \  \  \  \   \  \  \  \  \  \  \  \  \  \  \  \  \  \  \  \  \  \  \   +(g+2) N_{g+2} \Big{(} (\frac{g+2}{2}+1) (c+d-2)(c+d-4)+(c-2 )(2d-1)+(d-2)(2c-1) \Big{)}.$$
\end{lemma}
\begin{proof}
We already computed the values of $|\Phi(g,4)| $ and $ |\Lambda(g,4)|$ in the 
case $1$ of the proof of Theorem \ref{s10}. Hence, we have 
$$ \  \  \  \  \   \  \  \  \  \  \  \  \  \  \  \ |\Phi(g+2,4)|=(g+2) N_{g+2} \Big{(}(g+2)( 
t_{c,4} +  
t_{d,4})+( b_{c,4}+b_{d,4})+\\
 2 \left({\begin{array}{*{20}{c}}
{\frac{g+2}{2}}\\
{2}

\end{array}}\right) [t_{c,2}^2+t_{d,2}^2]
$$ 

$$ \\ \   \  \  \  \  \  \  \  \  \  +
2(\frac{g+2}{2})^2 t_{c,2}t_{d,2}+
\frac{g+2}{2}(b_{c,2}+b_{d,2})(t_{c,2}+t_{d,2}) \Big{)},
$$
and
$$ \  \  \  \  \  \ \  \  \  |\Lambda(g+2,4)|= (g+2) N_{g+2} \Big{(} (\frac{g+2}{2}+1) \big{(}|N_1(v_0)|+|N_1(v_1)| \big{)}(t_{c,2}+ t_{d,2})
$$
$$
 \  \  \  \  \ \  \  \  \  \  \  + |N_1(v_0)|b_{d,2}+ 
|N_1(v_1)|b_{c,2}+N_{2}(v_0)+N_{2}(v_1) 
\Big{)}.$$
\end{proof}
\begin{lemma}
Let $G$ be a $(c, d)$-regular graph. Then the number of CWWCs of length $g+6$ with cycle of length $g+4$ is equal to 
$$(g+4) N_{g+4} (g+6)(c+d-2).$$
\end{lemma}
\begin{proof}
Since the values $|\Phi(g,2)| $ and $ |\Lambda(g,2)|$ are known by the proof of Theorem \ref{s5}. Thus, 
$$ \  \  \  \  \  \  \  \  \  \  \  \  \  \  |\Phi(g+4,2)|= (g+4) N_{g+4} \Big{(}(g+4)(t_{c,2}+t_{d,2}) +b_{c,2}+b_{d,2} \Big{)},$$ 
and 
$$|\Lambda(g+4,2)|=(g+4) N_{g+4} \Big{(}|N_1(v_0)|+|N_1(v_1)|\Big{)}.$$
\end{proof}


\begin{thebibliography}{}
\bibitem{blake} I. Blake And S. Lin, On Short Cycle Enumeration in Biregular Bipartite Graphs, IEEE Trans. Inf. Theory, Dec. 2017,
available online at: http://ieeexplore.ieee.org/document/8225637/.

\bibitem{chugg} K. M. Chugg, A. Anastasopoulos, and X. Chen, Iterative Detection:
Adaptivity, Complexity Reduction, and Applications. Norwell, MA:
Kluwer, 2001. 


\bibitem{cvet} D. Cvetkovi´c, P. Rowlinson, and S. Simi´c, An introduction to the
theory of graph spectra, London Mathematical Society, Student
Texts, vol. 75, 2010.

\bibitem{dehghan} A. Dehghan And A. Banihashemi,  On Computing the Multiplicity of Short Cycles in Bipartite Graphs Using the Degree Distribution and the Spectrum of the Graph, arXiv preprint arXiv:1806.01433, Jun. 2018.

\bibitem{bani} A. Dehghan And A. Banihashemi, On the Tanner Graph Cycle Distribution of
Random LDPC, Random Protograph-Based
LDPC, and Random Quasi-Cyclic LDPC Code
Ensembles, IEEE Trans. Inf. Theory, vol. 64, no. 6, pp. 4438-4451, June 2018.

\bibitem{flum} J. Flum And M. Grohe, “The parameterized complexity of 
counting problems,” SIAM J. Comput., vol. 33, no. 4, pp. 892-922,
2004.

\bibitem{hal} T. R. Halford and K. M. Chugg, “An algorithm for counting short cycles in bipartite graphs,” IEEE Trans. Inform. Theory,
vol. 52, no. 1, pp. 287-292, Jan. 2006.

\bibitem{hu} X.-Y. Hu, E. Eleftheriou, And D. M. Arnold, Regular and irregular 
progressive edge-growth Tanner graphs, IEEE Trans. Inform. Theory,
vol. 51, no. 1, pp. 386-398, Jan. 2005.

\bibitem{karimi} M. Karimi And A. H. Banihashemi, On the girth of quasi-cyclic protograph LDPC codes, IEEE Trans. Inform. Theory,
vol. 59, no. 7, pp. 4542-4552, July 2013.

\bibitem{kar} M. Karimi and A. H. Banihashemi, “Message-passing algorithms for counting short cycles in a graph,” IEEE Trans.
Communications, vol. 61, no. 2, pp. 485-495, Feb. 2013. 

\bibitem{mao} Y. Mao And A. H. Banihashemi, A heuristic search for good low-density
parity-check codes at short block lengths, in Proc. Int. Conf. Communications,
vol. 1, Helsinki, Finland, Jun. 2001, pp. 41-44.

\bibitem{richard} T. Richardson, M. A. Shokrollahi, and R. Urbanke, “Design of capacity approaching
irregular low-density parity check codes,” IEEE Trans. Inform.
Theory, vol. 47, no. 2, pp. 619-637, Feb. 2001.

\bibitem{tanner} R. M. Tanner, A recursive approach to low-complexity codes, IEEE Trans. Inform. Theory, vol. 27,  no. 5,
 pp. 533-547, 
Sep 1981.

\bibitem{xiao} H. Xiao And A. H. Banihashemi, “Error rate estimation of low-density parity-check codes on binary symmetric channels using cycle
enumeration,” IEEE Trans. Communications, vol. 57, no. 6, pp. 1550-1555, June 2009. 

\end{thebibliography}
\end{document}